\documentclass[12pt,reqno]{amsart}

\newcommand\version{September 6, 2011}

\usepackage{amsmath,amsfonts,amsthm,amssymb,amsxtra}

\newtheorem{theorem}{Theorem}[section]

\newtheorem{lemma}[theorem]{Lemma}

\theoremstyle{definition}

\theoremstyle{remark}

\numberwithin{equation}{section}

\renewcommand{\epsilon}{\varepsilon}

\newcommand{\N}{\mathbb{N}}

\renewcommand{\phi}{\varphi}
\newcommand{\R}{\mathbb{R}}

\newcommand{\Z}{\mathbb{Z}}

\DeclareMathOperator{\ess}{ess-sup}

\DeclareMathOperator{\tr}{Tr}

\begin{document}

\title[Entropy and the uncertainty principle ---
\version]{Entropy and the uncertainty principle}

\author{Rupert L. Frank}
\address{Rupert L. Frank, Department of Mathematics,
Princeton University, Washington Road, Princeton, NJ 08544, USA}
\email{rlfrank@math.princeton.edu}

\author{Elliott H. Lieb}
\address{Elliott H. Lieb, Departments of Mathematics and Physics,
Princeton University, P.~O.~Box 708, Princeton, NJ 08544,
USA}
\email{lieb@princeton.edu}

\thanks{\copyright\, 2011 by the authors. This paper may be
reproduced, in its entirety, for non-commercial purposes.\\
Work partially supported by NSF grants PHY--1068285
(R.L.F.) and PHY--0965859 (E.H.L.)}

\begin{abstract}
We generalize, improve and unify theorems
of Rumin, and Maassen--Uffink about classical entropies
associated to quantum density matrices. These theorems refer
to the classical entropies of the diagonals of a density
matrix in two different bases. Thus they provide a kind of
uncertainty principle. Our inequalities are sharp because
they are exact in the high-temperature or semi-classical
limit.
\end{abstract}

\maketitle

\section{Introduction}

The von Neumann entropy of a quantum state (density matrix)
can be calculated either in momentum space or in
configuration space and the two are equal. They can even be
zero. Nevertheless, the corresponding classical entropies,
determined by the diagonals of the two representations of
the density matrix, can be different, and they can even be
negative, but their sum cannot be arbitrarily small. This
sum of the classical entropies can thus serve as a measure
of the quantum mechanical uncertainty principle.

This point of view was advocated by Deutsch \cite{De}, who,
among other things, proved a lower bound on this sum, which
was later improved by Maassen and Uffink \cite{MaUf},
following a conjecture of Kraus \cite{Kr}. These
inequalities were obtained for a general pair of bases, not
just momentum and configuration space. In the
momentum--configuration basis an improvement on these
previous inequalities was made by Rumin \cite{Ru}, who was
able to add a term to the inequality involving the largest
eigenvalue of the density matrix. He raised the question
whether this additional term could be further improved by
using a larger quantity, namely, the von Neumann entropy of
the density matrix. In this paper we prove that this
surmise is correct.

We prove even more by combining the Maaseen-Uffink
investigation with the Rumin surmise. Rumin was concerned
with the momentum--configuration space duality, whereas
Maassen-Uffink were concerned with arbitrary pairs of bases of
the Hilbert space. For this they introduced a parameter $c$
which somehow quantifies the disparity between the two
bases. As one might expect, the $k,x$ pair has the largest
$c$-value, i.e., $c=1$. We show how our theorem applies to
any pair with the corresponding $c$-dependent improvement
found in \cite{MaUf}.

Our theorem and simple proof are supported by a
semi-classical intuition, as evidenced by our use of the
Golden-Thompson inequality. The only other ingredient in
our proof is the Gibbs variational principle. Because our
constant in Theorem~\ref{main} agrees with the
semi-classical limit it is the best possible.

%%%%%%%%%%%%%%%%%%%%%%%%%%%%%%%%%%%%%%%%%%%%%%%%%%%%%%

\section{Rumin's conjecture and its generalizations}

For any trace class operator $\gamma\geq 0$ on $L^2(\R^d)$ we denote by
$\rho_\gamma(x)=\gamma(x,x)$ its density; see \eqref{eq:density} for a precise
definition. Moreover,
$$
\hat\gamma(k,k') = \iint_{\R^d\times\R^d} e^{2\pi i(k\cdot
x-k'\cdot x')} \gamma(x,x')\,dx\,dx'
$$
and
$$
\rho_{\widehat\gamma}(k) = \hat\gamma(k,k) = 
\iint_{\R^d\times\R^d} e^{2\pi i k\cdot (x-x')} \gamma(x,x')
\,dx\,dx'\,.
$$
We note that if $\tr\gamma=1$, then
\begin{equation*}
 \int_{\R^d} \rho_\gamma(x) \,dx 
= \int_{\R^d} \rho_{\widehat\gamma}(k) \,dk = 1 \,.
\end{equation*}
Our main result is

\begin{theorem}\label{main}
 For any $\gamma\geq 0$ with $\tr\gamma=1$ and
$$
\int_{\R^d} \rho_\gamma(x) \ln_+ \rho_\gamma(x) \,dx
<\infty
\quad\text{and}\quad
\int_{\R^d} \rho_{\widehat\gamma}(k)
\ln_+\rho_{\widehat\gamma}(k)
\,dk < \infty \,,
$$
where $\ln_+\rho=\max\{\ln\rho,0\}$, one has
\begin{equation}
 \label{eq:main}
- \int_{\R^d} \rho_\gamma(x)\ln\rho_\gamma(x) \,dx 
- \int_{\R^d}
\rho_{\widehat\gamma}(k)\ln\rho_{\widehat\gamma}(k)
\,dk \geq -\tr\gamma\ln\gamma \,.
\end{equation}
\end{theorem}

\noindent\emph{Remarks.}
(1) While the entropy on the right side of \eqref{eq:main} is necessarily
non-negative, those on the left side can have either sign.\\
(2) Inequality \eqref{eq:main} is saturated in the semi-classical
limit. This can be verified by taking $\gamma=Z_\beta^{-1}
\exp(-\beta(-\Delta+x^2))$ and letting $\beta\to 0$; see
\cite{Ru}.\\
(3) For $\gamma$ of rank one, this is Hirschman's
inequality \cite{Hi}. This was improved by Beckner
\cite{Be}. However, because of (1) this improvement is not
possible if one allows for mixed states (i.e., $\gamma$ of
higher rank).\\
(4) The inequality for $\gamma$ equal to a multiple of a
projection was proved in \cite{Ru}.  More generally, Rumin proves \eqref{eq:main} with $\ln \|\gamma\|_\infty$ instead of $\tr\gamma\ln\gamma$.\\ 
(5) The inequality shares the following tensorization
property: If $d=n+m$, we can think of $L^2(\R^d)$ as
$L^2(\R^n)\otimes L^2(\R^m)$. Then the main inequality for
$\gamma=\gamma_n\otimes\gamma_m$ equals the sum
of the inequalities for $\gamma_n$ and $\gamma_m$.\\
(6) If, instead, we define the Fourier transform by
$$
\widetilde\gamma(p,q) = \iint e^{i(p\cdot x-q\cdot y)}
\gamma(x,y) \,\frac{dx\,dy}{(2\pi)^d} \,
\quad \text{and}\quad
\rho_{\widetilde\gamma}(p) = 
\iint e^{i p\cdot (x-y)} \gamma(x,y)
\frac{dx\,dy}{(2\pi)^d} \,,
$$
then \eqref{eq:main} becomes
$$
- \int \rho_\gamma(x)\ln\rho_\gamma(x) \,dx - \int
\rho_{\widetilde\gamma}(p)\ln\rho_{\widetilde\gamma}(p) \,dp
\geq -\tr\gamma\ln\gamma +  d\ln(2\pi) \,.
$$

\medskip

Theorem \ref{main} is a special case of a more general
Theorem \ref{main3} below. We listed Theorem~\ref{main}
separately because it was the starting point of our
investigation and was conjectured by Rumin.

The more general theorem includes the discrete case as well
as the continuous case in Theorem \ref{main}. It is not
entirely a triviality that the discrete and continuous
cases are contained in one theorem because, as is well
known, many entropy inequalities are true in one case and
not in the other. For example, the discrete entropy is
always positive while the continuous entropy can be, and
often is, negative.

The general set-up consists of two sigma-finite measure
spaces $(X,\mu)$ and $(Y,\nu)$. We denote by $L^2(X)$ and
$L^2(Y)$ the corresponding spaces of square-integrable
functions. Let $\gamma$ be a non-negative operator on
$L^2(X)$ with
$\tr\gamma=1$. Then we have $\gamma=\sum_j \lambda_j
|f_j\rangle\langle f_j|$
with orthonormal functions $(f_j)$ and numbers
$\lambda_j\in[0,1]$ satisfying
$\sum_j \lambda_j = 1$. We define the density $\rho_\gamma$
of $\gamma$, a function on $X$, by
\begin{equation}
 \label{eq:density}
\rho_\gamma(x) = \sum_j \lambda_j |f_j(x)|^2 \,.
\end{equation}
By monotone convergence, we have
\begin{equation}
 \label{eq:norm}
\int_X \rho_\gamma(x) \,d\mu(x) = \sum_j \lambda_j =
\tr\gamma =1 \,.
\end{equation}

Assume now that there is a unitary operator $\mathcal U:
L^2(X)\to L^2(Y)$. For $\gamma$ as before,
we define an operator $\hat\gamma$ on $L^2(Y)$ by
$$
\hat\gamma = \mathcal U \, \gamma \, \mathcal U^* \,.
$$
This operator is non-negative and has $\tr\hat\gamma=1$. Its
density $\rho_{\widehat\gamma}$ is defined similarly to
that of $\rho_\gamma$, namely,
$$
\rho_{\widehat\gamma}(y) = \sum_j \lambda_j |g_j(y)|^2 \,,
$$
where $\hat\gamma=\sum_j \lambda_j |g_j\rangle\langle g_j|$
and $g_j=\mathcal U f_j$. As in \eqref{eq:norm},
\begin{equation}
 \label{eq:normy}
\int_Y \rho_{\widehat\gamma}(y) \,d\nu(y) = 1 \,.
\end{equation}

Our final assumption is that $\mathcal U$ is bounded from
$L^1(X)$ to
$L^\infty(Y)$. This property guarantees that $\mathcal U$
has an integral kernel $\mathcal U(y,x)$ with
$$
\infty>\!\|\mathcal U\|_{L^1\to L^\infty}\! = \ess_{x,y}
\!|\,\mathcal U(y,x)| := \sup\left\{ t: (\mu\times\nu)(\{
(x,y) : |\,\mathcal U(x,y)|>t \}) >0 \right\}. 
$$

\begin{theorem}\label{main3}
Under the above assumptions, let $\gamma\geq 0$ be an
operator in $L^2(X)$ with
$\tr\gamma=1$ and such that
$$
\int_X \rho_\gamma(x) \ln_+ \rho_\gamma(x) \,d\mu(x) <\infty
\quad\text{and}\quad
\int_Y \rho_{\widehat\gamma}(y)
\ln_+\rho_{\widehat\gamma}(y)
\,d\nu(y) < \infty \,,
$$
where $\ln_+\rho=\max\{\ln\rho,0\}$. Then
\begin{equation}
 \label{eq:main3}
-\int_X \rho_\gamma(x) \ln\rho_\gamma(x) \,d\mu(x) 
-\int_Y \rho_{\widehat\gamma}(y) \ln\rho_{\widehat\gamma}(y)
\,d\nu(y) 
\geq - \tr\gamma\ln\gamma -2\ln\|\mathcal U\|_{L^1\to
L^\infty} \,.
\end{equation}
\end{theorem}

We illustrate this theorem by some \textbf{examples}.
%\noindent\textbf{Examples.}
\begin{enumerate}
 \item If $X=Y=\R^d$ with Lebesgue measure and $\mathcal U$
the Fourier transform (i.e., $\mathcal U(k,x)=e^{-2\pi i
k\cdot x}$), then we recover Theorem \ref{main}. In this
case, $-2\ln\|\mathcal U\|_{L^1\to L^\infty} =0$.
\smallskip
\item Let $X=(-L/2,L/2)$ with Lebesgue measure, $Y=L^{-1}
\Z$ with $L^{-1}$ times counting measure and let $\mathcal U$ be the
discrete Fourier transform, that is, $\mathcal U(k,x) =
e^{-2\pi i k x}$. Then \eqref{eq:main3} holds
with $-2\ln\|\mathcal U\|_{L^1\to L^\infty} =0$.
\smallskip
\item Let $X=Y=\Z/N\Z=\{0,1,\ldots,N-1\}$ for some $N\in\N$
with counting measure and let $\mathcal U(k,n)=N^{-1/2}
e^{-i2\pi kn/N}$. Then
\eqref{eq:main3} holds with $-2\ln\|\mathcal U\|_{L^1\to
L^\infty} =\ln N$.
\smallskip
\item The following is a generalization of Example (3) and
is related to \cite{De,Kr,MaUf}. Let $(|a_j\rangle)_j$ and
$(|b_k\rangle)_k$ two orthonormal bases in a separable
Hilbert space $\mathcal H$ and put
$$
c=\sup_{j,k} |\langle a_j|b_k\rangle| \,.
$$
By the Schwarz inequality, $0<c\leq 1$. Let $\gamma\geq 0$
be an operator on $\mathcal H$ with $\tr\gamma=1$. Define
$$
p_j := \langle a_j|\gamma|a_j \rangle \,,
\qquad
q_k := \langle b_k|\gamma|b_k \rangle \,.
$$
Then
\begin{equation}
 \label{eq:mauf}
 -\sum_j p_j \ln p_j - \sum_k q_k \ln q_k \geq -
\tr\gamma\ln\gamma -2 \ln c \,,
\end{equation}
which follows from Theorem \ref{main3} by noting that, if
the change of bases is denoted by $\mathcal U$, then
$\|\,\mathcal U\|_{L^1\to L^\infty} = c$. The weaker
inequality without the term $\tr\gamma\ln\gamma$ on the
right side was shown in \cite{MaUf} with a different
proof.\\
In passing, we note that each of
the entropies on the left side of \eqref{eq:mauf} is
greater than or equal to $-\tr\gamma\ln\gamma$. This
follows from the concavity of $-p\ln p$, the fact (derived
from the variational principle) that the sequence $(p_j)$ is
majorized by the sequence of eigenvalues of $\gamma$, and
Karamata's theorem (see, e.g., \cite{HaLiPo} or \cite[Rem.
4.7]{LiSe}).
\end{enumerate}

%%%%%%%%%%%%%%%%%%%%%%%%%%%%%%%%%%%%%%%%%%%%%%%%%%%%%

\section{Proof of Theorem \ref{main3}}

Our proof is based on the following two well known lemmas in
quantum statistical mechanics; see, e.g., \cite{Ca,Si}.

\begin{lemma}[Gibbs variational principle]\label{gibbs}
Let $H$ be a self-adjoint operator such that $e^{- H}$ is
trace class. Then for any $\gamma\geq 0$ with $\tr\gamma=1$,
\begin{align*}
\tr \gamma H + \tr\gamma\ln\gamma
\geq - \ln\tr e^{-H}
\end{align*}
with equality iff $\gamma= \exp(- H)/\tr\exp(- H)$.
\end{lemma}

\begin{lemma}[Golden-Thompson inequality]\label{gt}
For self-adjoint operators $A$ and $B$,
$$
\tr e^{A+B} \leq \tr e^{A/2} e^B e^{A/2} \,.
$$
\end{lemma}

\begin{proof}[Proof of Theorem \ref{main3}]
We first note that
$$
- \int \rho_\gamma(x)\ln\rho_\gamma(x) \,d\mu(x)
- \int \rho_{\widehat\gamma}(y)\ln\rho_{\widehat\gamma}(y)
\,d\nu(y)
= \tr \gamma H
$$
with the operator $H=-\ln\rho_\gamma 
-\mathcal U^* \ln\rho_{\widehat\gamma} \,\mathcal U$ in
$L^2(X)$. Here, $\ln\rho_\gamma$ and
$\ln\rho_{\widehat\gamma}$ are considered as multiplication
operators, and we used the fact that $\tr_{L^2(X)} \mathcal
U^* A \,\mathcal U = \tr_{L^2(Y)} A$. By Lemmas \ref{gibbs}
and \ref{gt},
\begin{align*}
- \! \int \rho_\gamma(x)\ln\rho_\gamma(x) \,d\mu(x) 
- \! \int
\rho_{\widehat\gamma}(y)\ln\rho_{\widehat\gamma}(y)\,d\nu(y)
+ \tr\gamma\ln\gamma 
& \geq -\ln\tr e^{-H} \\
& \geq -\ln \tr \rho_\gamma^{1/2} \mathcal U^*
\rho_{\widehat\gamma} \,\mathcal U \rho_\gamma^{1/2} \,\,.
\end{align*}
The trace on the right side is the square of the
Hilbert-Schmidt norm of the operator
$\rho_{\widehat\gamma}^{1/2} \,\mathcal U
\rho_\gamma^{1/2}$, which has kernel
$$
\rho_{\widehat\gamma}(y)^{1/2} \,\mathcal U(y,x)
\rho_\gamma(x)^{1/2} \,.
$$
Thus,
\begin{align*}
 \tr \rho_\gamma^{1/2} \mathcal U^*
\rho_{\widehat\gamma} \,\mathcal U \rho_\gamma^{1/2}
& = \iint_{X\times Y} \rho_{\widehat\gamma}(y)
|\,\mathcal U(y,x)|^2 \rho_\gamma(x) \,d\mu(x)\,d\nu(y) \\
& \leq \|\mathcal U\|_{L^1\to L^\infty}^2 
\int_Y \rho_{\widehat\gamma}(y) \,d\nu(y)\
\int_X \rho_\gamma(x) \,d\mu(x) \,.
\end{align*}
By \eqref{eq:norm} and \eqref{eq:normy}, this equals
$\|\mathcal U\|_{L^1\to L^\infty}^2$, and the proof of the
theorem is complete.
\end{proof}

%%%%%%%%%%%%%%%%%%%%%%%%%%%%%%%%%%%%%%%%%%%%%%%%%%%%%%%%%%%%

%%%%%%%%%%%%%%%%%%%%%%%%%%%%%%%%%%%%%%%%%%%%%%%%%%%%%%%%%%%%

\bibliographystyle{amsalpha}

\begin{thebibliography}{HaLiPo}

\bibitem[Be]{Be} W. Beckner, \textit{Inequalities in Fourier
analysis}. Ann. of Math. (2) \textbf{102} (1975), no. 1,
159--182.
\bibitem[Ca]{Ca} E. A. Carlen, \textit{Trace inequalities
and quantum entropy. An introductory course}. In: Entropy
and the quantum, R. Sims and D. Ueltschi (eds.), 73--140,
Contemp. Math. \textbf{529}, Amer. Math. Soc., Providence,
RI, 2010.
\bibitem[De]{De} D. Deutsch, \textit{Uncertainty in quantum
measurements}. Phys. Rev. Lett. \textbf{50} (1983),
631--633.
\bibitem[HaLiPo]{HaLiPo} G. H. Hardy, J. E. Littlewood, G.
P\'olya, \textit{Inequalities}. Cambridge University Press,
Cambridge, 1952.
\bibitem[Hi]{Hi} I. I. Hirschman Jr., \textit{A note on
entropy}. Amer. J. Math. \textbf{79} (1957), 152--156.
\bibitem[Kr]{Kr} K. Kraus, \textit{Complementary observables
and uncertainty relations}. Phys. Rev. D \textbf{35} (1987),
3070--3075.
\bibitem[LiSe]{LiSe} E. H. Lieb, R. Seiringer, \textit{The
stability of matter in quantum mechanics}. Cambridge
University Press, Cambridge, 2010. 
\bibitem[MaUf]{MaUf} H. Maassen, J. B. M. Uffink,
\textit{Generalized entropic
uncertainty relations}. Phys. Rev. Lett. \textbf{60} (1988), no. 12, p.
1103--1106.
\bibitem[Ru]{Ru} M. Rumin, \textit{Balanced distribution-energy inequalities
and related entropy bounds}. Preprint (2010), arxiv:1008.1674
\bibitem[Si]{Si} B. Simon, \textit{Trace ideals and their
applications}. Second edition. Amer. Math. Soc., Providence,
RI, 2005.

\end{thebibliography}

\end{document}